\documentclass[11pt]{article}
\usepackage{fullpage}
\usepackage{times}
\usepackage{subfigure}
\usepackage{color}
\usepackage{url}
\usepackage{graphicx}
\usepackage{amsmath}
\usepackage{amssymb}
\usepackage{amsthm}
\usepackage{microtype}

\usepackage{algorithmic}
\usepackage[ruled]{algorithm2e}

\newcommand{\etal}{{\em et al.}}
\newcommand{\expect}{\textbf{E}}
\newcommand{\prob}{{\bf \mbox{\bf Pr}}}

\definecolor{gray}{rgb}{0.5,0.5,0.5}
\newcommand{\e}{{\epsilon}}

\newtheorem{theorem}{Theorem}
\newtheorem{lemma}[theorem]{Lemma}

\newtheorem{corollary}[theorem]{Corollary}
\newtheorem{definition}[theorem]{Definition}
\newtheorem{remark}[theorem]{Remark}

\usepackage{hyperref}
\hypersetup{
	bookmarks=true,
	bookmarksnumbered=true,
	bookmarksopen=true,
	pdftitle={Single pass sparsification in the streaming model with edge deletions},
	pdfauthor={Ashish Goel, Michael Kapralov, Ian Post}
}

\begin{document}

\title{\Large Single pass sparsification in the streaming model with edge deletions}
\setcounter{page}{0}
\author{Ashish Goel\thanks{
    Departments of Management Science and Engineering and (by courtesy)
    Computer Science, Stanford University.
    Email: \href{mailto:ashishg@stanford.edu}{ashishg@stanford.edu}.
    Research supported in part by  NSF grants 0915040 and 0904314.}\\    
\and Michael Kapralov\thanks{
    Institute for Computational and Mathematical Engineering, Stanford University.
    Email: \href{mailto:kapralov@stanford.edu}{kapralov@stanford.edu}. Research supported by NSF grant 0904314 and a Stanford Graduate Fellowship.}\\   
\and Ian Post\thanks{Department of Computer Science, Stanford University,
Email: \href{mailto:itp@stanford.edu}{itp@stanford.edu}. Research supported by NSF grants 0915040 and 0904314.}
}

\maketitle

\begin{abstract}
In this paper we give a construction of cut sparsifiers of Bencz\'{u}r and Karger in the {\em dynamic} streaming setting in a single pass over the data stream. Previous constructions either required multiple passes or were unable to handle edge deletions. We use $\tilde{O}(1/\e^2)$ time for each stream update and $\tilde{O}(n/\e^2)$ time to construct a sparsifier.  Our $\e$-sparsifiers have $O(n\log^3 n/\e^2)$ edges. The main tools behind our result are an application of sketching techniques of Ahn \etal[SODA'12] to estimate edge connectivity together with a novel application of sampling with limited independence and sparse recovery to produce the edges of the sparsifier.
\end{abstract}

\section{Introduction}

We study the problem of graph sparsification on dynamic graph streams.
Graph sparsification was introduced by Bencz\'{u}r and Karger \cite{benczurkarger96},
who gave a near linear time procedure that takes as input an
undirected graph $G$ on $n$ vertices and constructs a weighted subgraph $H$ of
$G$ with $O(n\log n/\e^2)$ edges such that the value of every cut in $H$ is
within a $1\pm \e$ factor of the value of the corresponding cut in $G$. 
This
algorithm has subsequently been used to speed up algorithms for a
host of applications involving cuts and flows such as finding
approximately minimum or sparsest cuts in graphs (\cite{benczurkarger96,
  krv06}) as well as other applications (e.g.\ \cite{kl02}). 
  Spielman and Teng introduced a stronger class of sparsifiers called spectral sparsifiers \cite{spieltengspectral}.
Subsequent work has developed a number of efficient algorithms for constructing cut and spectral sparsifiers \cite{benczurkarger96, ss:sample2008, bss09, kmst10, FHHP11, kp12}. 

The algorithms developed in
\cite{benczurkarger96, ss:sample2008, FHHP11, kp12} take near-linear time in the
size of the graph and produce very high quality sparsifiers, but require
random access to the edges of the input graph $G$, which is often
prohibitively expensive in applications involving massive data sets. The
streaming model of computation, which restricts algorithms to use a small
number of passes over the input and space polylogarithmic in the size of the
input, has been studied extensively in various application domains---see \cite{b:streaming} for an overview---but has proven too restrictive for even the
simplest graph algorithms. Even testing $s-t$ connectivity requires
$\Omega(n)$ space \cite{henz:lb}. The less restrictive semi-streaming model, in which the
algorithm is restricted to use $\tilde O(n)$ space, is more suited for graph
algorithms~\cite{fkmsz05, mcgregor:stream}. 

The problem of constructing graph sparsifiers in the semi-streaming model was first considered by  Ahn and Guha~\cite{anh-guha},
who gave a one-pass algorithm for finding Bencz\'{u}r-Karger type sparsifiers
with a slightly larger number of edges than the original Bencz\'{u}r-Karger
algorithm, $O(n\log n\log\frac{m}{n}/\e^2)$ as opposed to $O(n\log
n/\e^2)$ using $\tilde O(n)$ space.  Subsequently, \cite{kl11} obtained an algorithm for constructing stronger {\em spectral} sparsifiers of size $O(n\log n/\e^2)$ in a single pass in the streaming model. All of these algorithms work only in the {\em incremental model}, where edges can be added to the graph but not removed.

In a recent paper \cite{agm} Ahn, Guha and McGregor introduced a beautiful graph sketching approach to streaming computations in dynamic streams, i.e.\ allowing both edge additions and deletions. They showed that connectivity can be determined in $\tilde O(n)$ space in this setting, and gave a {\em multi-pass} algorithm for obtaining cut sparsifiers in small space. Their techniques center around the use of linear sketches, which have been heavily studied in the field of compressed sensing/sparse recovery originating in \cite{crt,donoho}.  See \cite{gilbertindyksurvey} for a survey.
%The focus of this paper is to obtain a single-pass algorithm for sparsification in dynamic streams. 
The focus of this paper is to provide a {\em single-pass} implementation of cut sparsification on dynamic streams in the semi-streaming model.

\if 0
Apart from the issue of random access vs disk, the semi-streaming model is
also important for scenarios where edges of the graph are
revealed one at a time by an external process. For example, this application maps
well to online social networks where edges arrive one by one, but efficient
network computations may be required at any time, making it particularly
useful to have a dynamically maintained sparsifier.
\fi

\paragraph{Our results:} Our main result is an algorithm for constructing cut sparsifiers in a single pass in dynamic streams with edge deletions. We prove
\begin{theorem}\label{thm:main}
There exists a single-pass streaming algorithm for constructing  an $\e$-cut sparsifier of an unweighted, undirected graph $G=(V, E)$ with $n$ vertices and $m$ edges in the dynamic model using $\tilde O(n/\e^2)$ space. The size of the sparsifier is $O(n\log^3 n/\e^2)$, and the runtime of the algorithm is $\tilde O(1/\e^2)$ per update.  At each point in the stream we can recover the edges of the sparsifier in time $\tilde O(n/\e^2)$.
\end{theorem}

Our sparsification algorithm works by sampling edges at a rate inversely proportional to their edge connectivity, which was shown to work in \cite{FHHP11}. In order to do this we maintain two sets of data structures. The first estimates connectivities, and the second does the actual sampling. We estimate connectivities by sampling edges of the input graph at a geometric sequence of sampling rates and recovering connected components of these samples using a result of \cite{agm}. The second set of data structures stores a linear sketch of the actual samples we use in our sparsifier, also sampling at a geometric sequence of rates. Using sparse recovery and the linearity of our sketch we are able to reconstruct the necessary samples when needed.

\noindent
\paragraph{Organization:} We start by giving preliminaries on graph sparsification in Section~\ref{sec:prelim}. We then describe the algorithm in Section~\ref{sec:unlim-indep} and Section~\ref{sec:lim-indep}. 
%Our algorithm maintains samples of the input graph at a geometric sequence of sampling rates in order to estimate connectivities, as well as maintains samples of the edges to be included in the final sparsifier. 
Maintaining our samples in the dynamic model requires knowing, for each edge in the graph, whether or not it was included in each sample. This can be easily achieved if we assume that our algorithm has access to $\tilde \Theta(n^2)$ independent random bits, which, however, is not feasible in $\tilde O(n)$ space. For simplicity of presentation, we first describe our algorithm assuming that it has access to $\tilde \Theta(n^2)$ random bits in Section~\ref{sec:unlim-indep}. We show how to obtain sufficiently good estimates of edge connectivities in a single pass, as well as recover the edges of a sparsifier using sparse recovery techniques.
In Section~\ref{sec:lim-indep} we show how to remove the assumption that the algorithm has access to $\tilde \Theta(n^2)$ independent random bits using random variables with limited independence, obtaining a $\tilde O(n/\e^2)$ space single-pass algorithm for sparsification in the dynamic model.

\section{Sparsification preliminaries}
\label{sec:prelim}

We will denote by $G(V, E)$ the undirected input graph with vertex set $V$ and edge set $E$ with $|V|=n$ and $|E|=m$. For our purposes $G$ will be unweighted, but the results in this section apply to weighted graphs.
For any $\e > 0$, we say that a weighted graph $G'(V,E')$ is an {\em $\e$-sparsification}
of $G$ if the weight of every cut in $G'$ is within $(1\pm \e)$ of the corresponding
cut in $G$.

%\subsection{Bencz\'{u}r-Karger Sampling Scheme}

Sparsification algorithms work by sampling edges with probabilities inversely proportional to some measure of connectivity.
The simplest of these is edge-connectivity:

\begin{definition}
A graph $G$ is \emph{$k$-connected} if the value of each cut in $G$ is at least $k$, and an edge $e$ has \emph{edge-connectivity} $c_e$ if $c_e$ is the value of the minimum cut separating its endpoints.
\end{definition}

Graphs with high $k$-connectivity are particularly simple to sample:

\begin{theorem}[\cite{kar94a}] \label{thm:karger}
Let $G=(V, E)$ be a $k$-connected graph on $n$ nodes, and let $G'$ be obtained from $G$ by sampling edges independently with probability $p=\Theta(\log n/(\e^2k))$, and giving sampled edges weight $1/p$. Then $G'$ is an $\e$-sparsifier of $G$ with high probability.
\end{theorem}

The Bencz\'{u}r-Karger algorithm samples according to a more strict notion of connectivity, referred to as
\emph{strong connectivity}, defined as follows:

\begin{definition}[\cite{benczurkarger96}]
A \emph{$k$-strong component} is a maximal $k$-connected vertex-induced subgraph.
The \emph{strong connectivity} of an edge $e$, denoted by $s_e$, is the largest $k$ such that a $k$-strong component contains $e$, and we say $e$ is \emph{$k$-strong} if its strong connectivity is $k$ or more, and \emph{$k$-weak} otherwise.
\end{definition}

\if 0 
Given any two collections of sets that partition $V$, say $S_1$ and $S_2$,
we call
$S_2$ is a {\em refinement} of $S_1$ if for any $X \in S_1$ and $Y \in S_2$,
either $X \cap Y = \emptyset$ or $Y \subset X$. In other words, $S_1 \cup S_2$
form a laminar set system.
Note that the set of $k$-strong components form a partition of the vertex set of $G$, and the set of $(k+1)$-strong components forms a refinement this partition.
\fi 
The following two lemmas will be useful in our analysis:

\begin{lemma}[\cite{benczurkarger96}]
\label{lm:k-weak}
The number of $k$-weak edges in a graph on $n$ vertices is bounded by $k(n-1)$.
\end{lemma}

\begin{lemma}[\cite{benczurkarger96}]
\label{lm:bounded}
Let $G=(V, E)$ denote an undirected graph on $n$ nodes. For an edge $e\in E$ let $s_e$ denote the strong connectivity of $e$. Then $\sum_{e\in E} 1/s_e\leq n-1$.
\end{lemma}

We also rely on Bencz\'{u}r and Karger's main result, which is as follows:

\begin{theorem}[\cite{benczurkarger96}] \label{thm:bk-sampling}
Let $G'$ be obtained by sampling edges of $G$ with probability $p_e=\min\{\rho/(\e^2 s_e), 1\}$, where $\rho=16(d+2)\ln n$, and giving each sampled edge weight $1/p_e$. Then
$G'$ is an $\e$-sparsification of $G$ with probability at least $1-n^{-d}$.
Moreover, the expected number of edges in $G'$ is $O(n \log n)$.
\end{theorem}

It follows easily from the proof of Theorem \ref{thm:bk-sampling} in \cite{benczurkarger96} that
if we over-sample by using an {\em underestimate} of edge strengths, the resulting graph is
still an $\e$-sparsification.

\begin{corollary}\label{cor:oversampling}
Let $G'$ be obtained by sampling each edge of $G$ with probability $\tilde p_e\geq p_e$ and and give every sampled edge $e$ weight $1/\tilde p_e$.
Then
$G'$ is an $\e$-sparsification of $G$ with probability at least $1-n^{-d}$.
\end{corollary}

%\subsection{Sampling using weak connectivities}

Recently Fung \etal \cite{FHHP11} proved that a more aggressive sampling method, namely sampling using edge connectivities as opposed to strong connectivities, also produces cut sparsifiers, and we will also require this result.

\begin{theorem}[\cite{FHHP11}]
\label{thm:wk-sampling}
Let $G'$ be obtained from a weighted graph $G$ by independently sampling edge $e$ with probability $p_e = \rho/c_e$, where $\rho= \Theta(\log^2 n/\e^2)$. Then, $G'$ contains $O(n \log^2 n/\e^2)$ edges in expectation and is an $\e$-sparsification whp.
\end{theorem}

\section{Sparsification with free randomness}
\label{sec:unlim-indep} 
In this section we present a dynamic sparsifier under the assumption that the algorithm has access to $\tilde{\Theta}(n^2)$ random words. We will remove this assumption in Section~\ref{sec:lim-indep}. 

Our input graph $G$ is undirected and unweighted.
As in \cite{agm}, we will use the following representation of $G$:

\begin{definition}
Given an unweighted graph $G=(V, E)$, let $A_G$ be the $n\times {n \choose 2}$ matrix with entry $(u, (v, w))\in [n]\times {[n] \choose 2}$ and $v < w$ given by 
\[
a_{u, (v, w)}=
\begin{cases}
1&\text{if $u=v$ and $(v, w)\in E$}\\
-1&\text{if $u=w$ and $(v, w)\in E$}\\
0&\text{otherwise}
\end{cases}
\]\end{definition}

Updates to the graph $G$ in the form of the addition or deletion of an edge arrive one at a time in a streaming fashion.  An update cannot delete an edge that does not exist or add one that already does, but other than these restrictions the order is adversarial, and the stream can be arbitrarily long.
We need to maintain a data structure using only $\tilde{O}(n)$ space that allows us to efficiently construct an $\e$-sparsifier of the current graph $G$ after any sequence of updates. We will accomplish this using a collection of linear sketches of the rows of $A_G$.

Our algorithm has two components: the first will maintain an estimate of the connectivity of each edge and therefore its sampling rate (as discussed in Section~\ref{sec:prelim}), and the second will store the actual samples.
The former uses the tools developed by Ahn \etal \cite{agm}, and the latter is based on the technique of sparse recovery developed in the sketching and compressed sensing literature \cite{gilbertindyksurvey}.

Before delving into the details, we elaborate on our use of randomness.
In this section we assume that for each pair $(u, v)\in [n]^2$ the algorithm has access to a uniformly random number $h_{(u, v)}\in [0, 1]$.  In fact, we will need $O(\log n)$ independent copies of these random numbers, which we will denote by $h^b_{(u,v)}$, $b=1,\ldots, O(\log n)$. These random variables will be used to estimate sampling rates for edges of $G$. We will also assume access to independent random numbers $g^b_{(u,v)}\in [0, 1]$, $b=1,\ldots, O(\log n)$, which we will use to determine a partition of the vertex set needed for sampling. Finally, we also assume access to independent random numbers $g^*_{(u, v)}\in [0, 1]$ which will be used to sample edges of the sparsifier.  
Once $g$, $g^*$ and $h$ are sampled, they are fixed for the duration of the algorithm. This is important for handling deletions, as it ensures that edges can be removed from exactly those sketches to which they have been added.

It will be convenient to think of all these numbers as independent in this section, even though this is not feasible in subquadratic space in the semi-streaming model.  In Section~\ref{sec:lim-indep} we will show that using numbers that are only $\tilde O(1)$-wise independent for fixed $u$ and independent for different $v$ is sufficient, leading to a space-efficient solution. 
Some of our subroutines will also require their own internal entropy, but the total used will be only $\tilde{O}(n)$ words.

\subsection{Estimating edge connectivity}

The building block of our connectivity estimates is the following result from \cite{agm} that finds connected components by sketching the rows of $A_G$:

\begin{theorem}[\cite{agm}]
\label{thm:stream-connectivity}
There is a single-pass, linear sketch-based algorithm supporting edge additions and deletions that uses $O(n\log^3 n)$ space and returns a spanning forest of the graph with high probability.
\end{theorem}

Our sketch is simple.
We consider random samples of the input graph at geometric sampling rates and find connected components in each sample. Specifically, for each $a=1,\ldots, O(\log n)$ denote  by $G^b_a=(V, E^b_a)$ a subsample of the edges of $G$ obtained by setting
\begin{equation}\label{eq:e-sample}
E^b_a=\left\{(u,v)\in E: \min\{h^b_{(u,v)}, h^b_{(v, u)}\}<2^{-a}\right\}.
\end{equation}
We maintain connectivity data structures $C^b_a$ for each of the subgraphs $G^b_a$ for $a=1,\ldots, O(\log n),$ $b=1,\ldots, O(\log n)$ using Theorem \ref{thm:stream-connectivity}.
It is important to note that use of the functions $h_{(u,v)}$ for sampling rather than fresh randomness allows us to handle deletions properly by deleting a removed edge only from the sketches that we added it to by simply sampling and using the consistent samples as input to the sketch in Theorem \ref{thm:stream-connectivity}.

\begin{remark}
Note that an edge $(u,v)$ is included in $E^b_a$ if the minimum of $h^b_{(u,v)}$ and $h^b_{(v, u)}$ is smaller than $2^{-a}$. This will be important for the proof of correctness for hash functions with limited independence in Section~\ref{sec:lim-indep}.
\end{remark}

The connectivity structure of the subgraphs $G^b_a$ allows us to associate a sampling rate with each edge.  For an edge $(u, v)\in E$ we use the smallest sampling rate $2^{-a}$ at which $u$ and $v$ are still in the same component as an estimate of the sampling rate for $(u, v)$. Independent repetition $O(\log n)$ times reduces the variance sufficiently to get precise estimates.

\newcommand{\V}{{\cal V}}

Define $\V_a$ as the partition of vertices in $V$ induced by the intersection of all partitions $C^b_a, b=1,\ldots, O(\log n)$. That is, vertices $u$ and $v$ are in the same connected component in $\V_a$ if and only if they are connected in $C^b_a$ for all $b=1,\ldots, O(\log n)$.
For an edge $(u, v)$ let $L(u, v)$---the {\em level of $(u,v)$}---denote the largest $a$ such that $u$ and $v$ are in the same component in $\V_a$, and for a vertex $v$ let the {\em level} $L(v)$ denote the largest $a$ such that $v$ is not a singleton in $\V_a$.

The level $L(e)$ of an edge serves as a proxy for its connectivity:

\begin{lemma}
\label{lm:level}
For all edges $e\in E$ one has 
\[
\Theta(s_e/\log n)\leq 2^{L(e)}\leq 2c_e
\]
with high probability, where $s_e$ denotes strong connectivity and $c_e$ denotes edge connectivity.
\end{lemma}
\begin{proof}
The first inequality follows from the fact that $\Theta(\log n)\cdot 2^{L}$-strongly connected components will stay connected with high probability by Theorem~\ref{thm:karger} when the functions $h^b_{(u, v)}$ used for sampling are truly random. Lemma~\ref{lm:t-wise-conn} from Section~\ref{sec:lim-indep} gives the result for hash functions with limited independence.

The second inequality follows by noting that  if there is a cut separating $u$ and $v$ of size at most $2^{L}/2$, then it will be empty with probability at least $1/2$ when we sample at rate $2^{-L}$ by Markov's inequality, so $u$ and $v$ will get disconnected in one of the $O(\log n)$ independent repetitions with high probability.
\end{proof}

This implies the levels can be used as sampling rates: 

\begin{lemma}
\label{lm:lvl-sampling}
Sampling edges independently at rate $p_e=O\left(\log^2 n/(\e^2 2^{L(e)})\right)$ and weighting sampled edges with $1/p_e$ produces a sparsifier with $O(n\log^3 n/\e^2)$ edges high probability. 
\end{lemma}
\begin{proof}
By Theorem \ref{thm:wk-sampling}
sampling at rate $O(\log^2 n/(\e^2c_e))$ works. By Lemma \ref{lm:level} $1/2^{L(e)} \ge 1/(2c_e)$, and oversampling only improves concentration. This proves the statement assuming that the sampling of edges is independent.  

The expected size of the sample is bounded by $\sum_{e\in E} p_e=O(\log^3 n/\e^2)\cdot \sum_{e\in E} 1/s_e=O(n\log^3 n/\e^2)$, where we used Lemma~\ref{lm:level} to bound $p_e=O(\log n)/s_e$ and the fact that $\sum_{e\in E}1/s_e\leq n-1$ by Lemma~\ref{lm:bounded}.
\end{proof}

\subsection{Maintaining edge samples}

We now show how to maintain small space sketches that will allow us to reconstruct the edges of the sparsifier. 
%Let $g_{(i, j)}$ denote uniformly random numbers in $[0, 1]$.
Our basic tool is the technique of sparse recovery from the field of compressed sensing \cite{gilbertindyksurvey}. A vector $A$ of dimension $N$ is {$k$-sparse} if it has at most $k$ non-zero entries, and a $k$-sparse, or approximately $k$-sparse, signal can be recovered with high probability from $O(k\log(N/k))$ non-adaptive linear measurements. Here we use the following result of Cormode and Muthukrishnan \cite{cormodemuthu} that allows recovery in $\tilde{O}(k)$ time at the cost of slightly sub-optimal sketch size and error guarantees:

\begin{theorem}[\cite{cormodemuthu}]
\label{thm:k-sparse}
We can construct a randomized $0/1$ matrix $T$ of dimension $O(ck\log^3n/\e^2) \times N$ such that for any $k$-sparse signal $A$ of dimension $N$, given the transformation $TA$ we can reconstruct $A$ exactly with probability at least $1-n^{-c}$ in time $O(c^2k\log^3 n/\e^2)$. The matrix $T$ is constructed using $O(1)$-wise independent hash functions, and individual entries can be queried efficiently.
\end{theorem}

We will also need the following result by Indyk \cite{indykstable} on sketching $\ell_1$ norms:

\begin{theorem}[\cite{indykstable}]
\label{thm:degree-sketch}
There is a linear sketch-based algorithm using $O(c\log^2 N/\e^2)$ space and $O(c\log^2 N/\e^2)$ random bits that can estimate the $\ell_1$ norm of a vector of dimension $N$ to within a factor of $(1\pm \e)$ with probability $1-N^{-c}$. The sketch can be updated in $O(\log N)$ time.
\end{theorem}

For a row $v$ of the matrix $A_G$ let $S^r_a(v)$ for $r=1,\ldots, O(\log n)$ denote linear sketches guaranteed by Theorem~\ref{thm:k-sparse} for $k=O(\log^3 n/\e^2)$ where $a$ corresponds to the geometric sequence of sampling rates, and $r=1,\ldots, O(\log n)$ are independent copies that are useful for recovery.  More precisely, $S^r_a(v)$ is a sketch of row $v$ in the matrix $A_{G_{a,r}'}$ where $G_{a, r}'$ has edges
\begin{equation}\label{eq:g-sample}
E'_{a, r}=\left\{(u,v)\in E: \min\{g^r_{(u,v)}, g^r_{(v, u)}\}<2^{-a}\right\}.
\end{equation}
Some of these sketches may accumulate more than $k$ edges and consequently cannot be decoded on their own, but we will prove this is not an issue.
We will also need sketches $d^r_a(v), r=1,\ldots, O(\log n)$ for the $\ell_1$-norm of the $v$-th row of the matrix $A_{G_{a, r}'}$. Note row $v$ of $A_{G_{a, r}'}$ contains a $\pm 1$ entry for each edge incident on $v$ in $G_{a, r}'$, so the $\ell_1$ norm corresponds exactly to the degree of $v$ in the sampled graph $G_{a, r}'$.
\begin{remark}
Note that we are using $O(\log n)$ independent samples $G'_{a, r}$ for each sampling rate $a$. This will be important in the proof of Lemma~\ref{lm:partition} below.
\end{remark}

Finally, we will also need another set of independent samples of $G$ that will be used to obtain the edges of the sparsifier. Let $G^*_{a}$ be the graph with edges
\begin{equation}\label{eq:g-hat-sample}
E^*_{a}=\left\{(u,v)\in E: \min\{g^*_{(u,v)}, g^*_{(v, u)}\}<2^{-a}\right\}.
\end{equation}
For each node $v$ and sampling rate $a$ we maintain sketches $S^*_a(v)$ of row $v$ in the matrix $A_{G^*_{a}}$ using Theorem~\ref{thm:k-sparse} for $k=O(\log^3 n/\e^2)$.
Here we do not need independent repetitions for each sampling rate $a$.

By the choice of the matrix $A_G$ and the linearity of the sketches, if ${\cal S} \subseteq V$ is a cut then $\sum_{v \in {\cal S}} d^r_a(v)$ is a sketch for the size of the cut and $\sum_{v \in {\cal S}} S^r_a(v)$ is sketch of a sample of its edges.  If $v$ is a supernode obtained by contracting a set of vertices ${\cal S}$, we write $d^r_a(v)$ to denote $\sum_{u \in {\cal S}} d^r_a(u)$ and similarly for $S^r_a(v)$ and $S^*_a(v)$.
For any fixed cut and fixed $G_a'$ the estimate given by $d^r_a(v)$ is close to expectation with high probability.

Before specifying the algorithm formally, we give the intuition behind it. Recall that for every vertex $u$ at level $L(u) =a$, we need to sample edges going from $u$ to vertices in $u$'s component in $\V_a$ with probability $\gamma \log^2 n/(\e^2 2^{a})$ for an appropriate constant $\gamma$. In order to do that, we will sample {\em all edges incident on $u$} with probability $\gamma \log^2 n/(\e^2 2^{a})$ and then throw away the ones that do not go to $u$'s component. In order to obtain such a sample, we will use the sketches $S^r_{a'}(u)$ that were made with sampling at rate $\gamma \log^2 n/(\e^2 2^{a})$. 

The main observation here is that if we contract connected components in $\V_{a+1}$ into supernodes, the resulting graph will have only $(\gamma \log n\cdot 2^{a+1})$-weak edges for a constant $\gamma$, so the
average degree will be no larger than $\gamma \log n\cdot 2^{a+1}$.  By repeatedly removing vertices with degree at most twice the average, nodes of this subgraph can be partitioned into sets $W_1, \ldots, W_t$ such that for each $i=1,\ldots, t$ and $u\in W_i$ the degree of $u$ in $W_{i}\cup\cdots\cup W_t$ is at most $4\gamma \log n 2^a$ and $t = O(\log n)$. Formally,

\begin{lemma} 
\label{lm:w-partition}
Suppose all edges in $G$ are $k$-weak. Then $V$ can be partitioned into $t=\log n$ sets $W_1,\ldots, W_t$ such that for all $v \in W_i$ the degree of $v$ when restricted to $W_i \cup \cdots \cup W_t$ is at most $2k$.
\end{lemma}

\begin{proof}
By Lemma~\ref{lm:k-weak} if $G$ has $n'$ nodes then it has at most $k(n'-1)$ edges. Let $W_1$ be the set of all nodes with degree at most $2k$. By Markov's inequality $W_1$ includes at least half the nodes. After removing $W_1$ and all incident edges we can repeat this process to find $W_2$, etc. At each iteration we remove at least half the nodes, so it terminates in $\log n$ iterations.
\end{proof}

This partition cannot actually be computed because we cannot properly update the degree sketches after removing $W_1$, but its existence allows us to prove that the same procedure works when using the lower degree sample $G_{a,r}'$.

Let $\Delta=\log(\gamma \log^2 n/\e^2)$. We need to use the samples $S^r_{a-\Delta}(u)$ for edges at level $a$.
We first bound the degrees in $G_{a-\Delta,r}'$:
\begin{lemma}\label{lm:low-degree}
Let $g^r_{(u, v)}$ be $\Theta(\log^3 n/\e^2)$-wise independent for fixed $u$, and independent for different $u$. Then the degree of all $u\in W_i$ in $G_{a-\Delta,r}'$ restricted to nodes $W_{i}\cup\ldots\cup W_t$ is at most $O(\log^3 n/\e^2)$ with high probability.
\end{lemma}
\begin{proof}
Consider a vertex $u\in W_i$ and let $N(u)$ denote the neighbors of $u$ in $W_{i}\cup\ldots\cup W_t$ in the full graph $G$. The size of its sampled neighborhood is bounded by 
\[
\sum_{v\in N(u)} \mathbf{1}_{g^r_{(u, v)}<\gamma \log^2 n/(\e^2 2^a)}+\sum_{v\in N(u)} \mathbf{1}_{g^r_{(v, u)}<\gamma \log^2 n/(\e^2 2^a)}
\]
The number of terms is $O(\log n 2^a)$.
The second sum consists of independent random variables, so standard Chernoff bounds apply. Concentration bounds from Theorem~\ref{thm:ch} apply to the first sum since they are sufficiently independent for expectation $O(\log^3 n/\e^2)$.
\end{proof}

Observe that if a node $u$ satisfies the bound in Lemma~\ref{lm:low-degree}, then the sketch $S^r_{a-\Delta}(u)$ can be decoded using Theorem~\ref{thm:k-sparse}.  Let $U_1$ be the set of decodable nodes.
We would like to argue that we can simply output a decoded edge $(u, v)$ as part of the sparsifier if and only if $g^r_{(u, v)}<\gamma \log^2 n/(\e^2 2^{a})$ and $v$ belongs to the same connected component as $u$ in $\V_a$.
Then, using the linearity of the sketches, we could subtract decoded edges of the form $(u, v), u\in U_1, v\in V \setminus U_1$ from the sketches $S(v)$ and $d(v)$, effectively removing $U_1$ from the graph, and then move on to $U_2$.

However, for technical reasons to avoid dependencies and ensure the algorithm works in small space, we cannot reuse the variables $g^r_{(u,v)}$. We need to calculate $U_2$ using the independent sketches $S^{r+1}$ and $d^{r+1}$ as opposed to $S^{r}$ and $d^r$ to avoid dependencies and also use the variables $g^*_{(u,v)}$ for the actual samples. The size of $r$ will remain bounded since we will prove the process terminates in $O(\log n)$ steps, but switching to $S^{r+1}$ introduces additional complications because to remove a vertex $u \in U_1$ from the graph, we must be able to recover the edges from $S^{r+1}_{a-\Delta}(u), \ldots, S^{O(\log n)}_{a-\Delta}(u)$. The following lemma accomplishes this:

\begin{lemma}
\label{lm:sr-decode}
Let $g^r_{(u, v)}$ be $\Theta(\log^3 n/\e^2)$-wise independent for fixed $u$, and independent for different $u$, and suppose the degree of $u$ in $G_{a-\Delta,r}'$ is at most $\alpha\log^3 n/\e^2$ where $\alpha$ is the constant from Lemma~\ref{lm:low-degree}.
Then the degree of $u$ in $G_{a-\Delta,r'}'$ is $O(\log^3 n/\e^2)$ for all $r' \ge r$ with high probability.
\end{lemma}

\begin{proof}
If the degree of $u$ in $G_{a-\Delta,r}'$ is at most $\alpha\log^3 n/\e^2$, we expect the degree in $G$ to be at most $(\alpha/\gamma)\log n 2^a$, and concentration inequalities for sampling with limited independence (Theorem~\ref{thm:ch}) show that with high probability its degree is at most, say, $2(\alpha/\gamma)\log n 2^a$.  Applying Theorem~\ref{thm:ch} again shows $u$'s degree in $G_{a-\Delta,r'}'$ is $O(\log^3 n/\e^2)$ for any $r'$, and taking a union bound over all $r'$ finishes the proof.
\end{proof}

We now state the algorithm formally. For each $a=1,\ldots, O(\log n)$ we denote the graph obtained by contracting all connected components in $\V_{a+1}$ into supernodes by $H_a$.
 
\begin{algorithm}[H]\label{alg:partition}
\caption{PARTITION(a)}
\begin{algorithmic}[1]
\STATE Let $H^1_{a} \leftarrow H_a$.
\FOR {$r\leftarrow 1$ to $O(\log n)$}
\STATE Estimate the degree of each $v \in H^r_{a}$ from sketches $d^r_{a-\Delta}(v)$
\STATE $U^r_{a}\leftarrow \{v\in H^r_{a}| d^r_{a-\Delta}(v) \leq 4\alpha\log^3 n/\e^2\}$ %\leq 2\gamma \log n\}$
\FOR {$u\in U^r_{a}$, $j=r+1,\ldots, O(\log n)$}
\STATE Run sparse recovery on $S^{j}_{a-\Delta}(u)$ \label{alg:partition:recovery}
%\STATE Output each recovered edge $(u, v), u\in U_r$ only if $g_{(u, v)}<\gamma \log^2 n/(\e^2 2^{a})$ and $L(u, v)=a$.
\STATE For all edges $(u,v)$ recovered from $S^{j}_{a-\Delta}(u)$, subtract $(u,v)$ from $S^{j}_{a-\Delta}(v)$ and $d^{j}_{a-\Delta}(v)$ %for all $v\in H^r_{a}\setminus U^r_{a}$.
\ENDFOR
\STATE $H^{r+1}_{a}\leftarrow H^r_{a}\setminus U^r_{a}$.
\ENDFOR
\RETURN $\{U^r_{a}\}_{r=1,\ldots, O(\log n)}$
\end{algorithmic}
\end{algorithm}
Here $\gamma$ is a constant such that sampling the edges of a $k$-connected graph at rate $\gamma \log n/k$ produces a connected subgraph with probability at least $1-n^{-10}$, and $\alpha$ is a constant bounding the degree in Lemma \ref{lm:low-degree}.

%
%In order to prove correctness of the algorithm, we will use 
%\begin{lemma}\label{lm:few-edges}
%With high probability for all $a=1,\ldots, O(\log n)$, all $r=1,\ldots, O(\log n)$
%\[
%\sum_{u\in V(H_r)} d_v\leq \gamma \log n 2^a |V(H_r)|.
%\]
%\end{lemma}
%\begin{proof}
%Consider all $\gamma \log n\cdot 2^{a+1}$-strongly connected components. They will stay connected after sampling at rate $2^{-(a+1)}$ with high probability, and hence will be contracted into supernodes in the definition of $H_1$, and hence all edges in $H_1$ are $k$-weak. By Lemma~\ref{lm:k-weak} the number of edges in $H_1$ is thus at most $\gamma \log n\cdot 2^a\cdot |V(H_r)|$.
%\end{proof}

We first prove
\begin{lemma}\label{lm:partition}
For all $a=1, \ldots, O(\log n)$, Algorithm~\ref{alg:partition} recovers a partition of $H_a$ such that for all $r=1,\ldots, O(\log n)$ for each $u\in U^r_{a}$ the degree of $u$ in $U^{r+1}_{a}\cup\ldots\cup U^{O(\log n)}_{a}$ in graph $G^*_{a-\Delta}$ is $O(\log^3 n/\e^2)$ with high probability.
\end{lemma}
\begin{proof}

We first prove that the constructed sets cover all of  $H_a$. Consider the set $U_1$.  By Lemma~\ref{lm:low-degree}, all $u \in W_1$ have degree $O(\log^3 n/\e^2)$ in $G_{a-\Delta, 1}'$ with high probability, so removing all nodes with degree at most $4\alpha\log^3 n/\e^2$ for large enough $\alpha$ will include all $u \in W_1$.  Lemma~\ref{lm:sr-decode} implies that the sparse recovery in line \ref{alg:partition:recovery} will succeed for all $j$ with high probability, so we can completely remove $U^r_a$ from the graph. Lemma~\ref{lm:sr-decode} also bounds the degree of $u \in U^r_a$ in graph $G^*_{a-\Delta}$, by replacing $G_{a-\Delta,r'}'$ with $G^*_{a-\Delta}$ in the statement of the lemma.

We now note that the identity of the set $U^r_{a}$ is independent of the randomness used for samples and sketches $S_{a-\Delta}^j$, $d_{a-\Delta}^j$, $j=r+1,\ldots, O(\log n)$. Thus, the same bounds on node degrees follow by a recursive application of the argument to $H_a\setminus U^1_a$. Furthermore, it follows by induction on $r$ that after removing $U^1_a,\ldots, U^r_a$ we have removed all of $W_1,\ldots, W_r$ with high probability, so the algorithm terminates in $O(\log n)$ iterations.
\end{proof}

Given the partition $U^1_a,\ldots, U^{O(\log n)}_a$, the algorithm for recovering the edges of the sparsifier is as follows:

\begin{algorithm}[H]\label{alg:rec-edges}
\caption{RECOVER(a)}
\begin{algorithmic}[1]
\FOR {$r=1,\ldots, O(\log n)$, $u\in U^r_{a}$}
\STATE Run sparse recovery on $S^*_{a-\Delta}(u)$
\STATE Output each recovered edge $(u, v), u\in U_r$ only if $g^*_{(u, v)}<\gamma \log^2 n/(\e^2 2^{a})$ and $L(u, v)=a$.
\STATE Subtract  recovered edges from $S^*_{a-\Delta}(v)$ for all $v\in U_{a, r+1}\cup \ldots U_{a, O(\log n)}$.
\ENDFOR
\end{algorithmic}
\end{algorithm}

We can now prove
\begin{theorem}
For each $v\in V(G)$ Algorithm~\ref{alg:rec-edges} recovers a sample of edges incident on  $v$, where edges are picked with probability $\gamma \log^2 n/(\e^2 2^{L(v)})$.
\end{theorem}
\begin{proof}

%It follows from Lemma~\ref{lm:few-edges} that Algorithm~\ref{alg:rec-edges} will terminate in $O(\log n)$ steps with high probability. Indeed, at each step $r$ the average degree of nodes in $H_r$ is at most $\gamma \log n 2^a |V(H_r)|$, so at least half the nodes will be removed in each step.
\if 0
With high probability we contract all $\gamma\log n 2^{a+1}$-strongly connected components, since they are connected  in $\V_{a+1}$, so all remaining edges are $\gamma\log n 2^{a+1}$-weak and by Lemma \ref{lm:k-weak} have average degree at most $\gamma\log n 2^{a+1}$. If we could compute $W_1, \ldots, W_r$ we would remove half the nodes at each iteration, proving that $r = O(\log n)$. 
\fi 
Note that sparse recovery succeeds whp by the degree bound in Lemma~\ref{lm:partition}. 
Finally, note that the structure of the partition $U_1\cup\ldots\cup U_r$ maps to each edge a single random variable $g_{(u, v)}$, so the probability of an edge being sampled is correct.
\end{proof}

We will need the following definition in Section~\ref{sec:lim-indep}:
\begin{definition}\label{def:control}
An edge $e=(u, v)\in E$ is controlled by a vertex $u\in V$ if $e$ is sampled using $g^*_{(u, v)}$.  We denote the set of edges controlled by $u$ by $E_u$.
\end{definition} 

We will also need
\begin{lemma}\label{lm:control-small}
Let $E^*$ be a set of edges.
For each $u\in V$ one has $\expect[|E^*\cap E_u|]=O(\log^4 n/\e^2)$.
\end{lemma}
\begin{proof}
Consider a vertex $u\in V$.  By Lemma~\ref{lm:partition} $u$ controls $O(\log^3 n 2^a/\e^2)$ edges at level $a$. Hence, the expected number of edges sampled at each level $a$ is $O(\log^3 n/\e^2)$. Hence, the expected number of edges controlled by $u$ across all levels is $O(\log^4 n/\e^2)$.
\end{proof}

\subsection{Runtime}

We now briefly summarize the time required to update the sketches and to construct a sparsifier. We do not optimize the log factors but only show updates require $\tilde{O}(1/\e^2)$ time and building a sparsifier requires $\tilde{O}(n/\e^2)$.  Each addition of deletion of an edge requires updating $C^b_a$, $S^r_a$ and $d^r_a$ for $a,b,r \le O(\log n)$. The sketches $C^b_a$ are built using $\ell_0$-samplers (see \cite{agm}) and can be updated in $\tilde{O}(1)$ time. For an edge $(u,v)$ we update all $O(\log n)$ copies of $S(u)$, $S(v)$, $d(u)$ and $d(v)$. By Theorems \ref{thm:k-sparse} and \ref{thm:degree-sketch}, these can each be updated in $\tilde{O}(1/\e^2)$ time. For $S(u)$ this is done by querying only the $\tilde{O}(1/\e^2)$ entries of the matrix $T$ we need.

Construction of the sparsifier is more expensive. If we query the sparsifier after each graph update it may need to be recomputed from scratch each time due to edge deletions, so we cannot amortize its cost across the updates. However, we will show it requires only $\tilde{O}(n/\e^2)$ time. Building all $O(\log n)$ $\V_a$ requires $\tilde{O}(n)$ operations each if $\ell_0$-sampling is done efficiently.  

Running one iteration of Algorithm~\ref{alg:partition} requires $O(n)$ estimations of $d$, $O(n)$ decodings of $S$, $\tilde{O}(k)$ updates to $S$ and $d$ for each of the $O(n)$ sparse recoveries and $O(n)$ additional bookkeeping. Since $S$ is $\tilde{O}(1/\e^2)$-sparse by Theorem~\ref{thm:k-sparse} decoding takes $\tilde{O}(1/\e^2)$ time. Summing over $\tilde{O}(1)$ values of $r$ and $a$, we use a total of $\tilde{O}(n/\e^2)$ time. Algorithm~\ref{alg:rec-edges} also does $O(n)$ sparse recoveries and $\tilde{O}(n/\e^2)$ updates to $S^*$ per iteration, which also totals to $\tilde{O}(n/\e^2)$ summing over all $r$ and $a$.

\subsection{Weighted graphs}
We note that even though we stated the algorithm for unweighted graphs, the following simple reduction yields a single pass dynamic sparsifier for weighted graphs, as long as when an edge is removed or updated, its weight is given together with the identity of its endpoints. Suppose that edge weights are integers between $1$ and $W$ (the general case can be reduced to this one with appropriate scaling and rounding).
Consider graphs $G_0,\ldots,G_{\log_2 W}$, where an edge $e=(u, v)$ belongs to the edge set of $G_b$ iff the binary expansion of $w_e$ has $1$ in position $b$, for $b=0,\ldots, \log_2 W$. Note that in order to preserve cuts in $G$ to a multiplicative factor of $1\pm \e$, it is sufficient to preserve cuts in each of $G_0,\ldots, G_b$ to the same factor. To do that, it is sufficient to maintain $\log_2 W$ copies of our algorithm operating on the graphs $G_b$ (this is feasible due to the assumption that edges are either added or completely removed, i.e.\ the weight of the removed edge  is given at the time of removal).
The space used and the number of edges in the sparsifier will both increase by a factor of $\log_2 W$.

\section{Sparsification with limited independence}\label{sec:lim-indep}
In this section we remove the assumption that the algorithm has access to $\tilde{\Theta}(n^2)$ bits of randomness by using  sampling with limited independence. We prove the following two statements. First, we show in Lemma~\ref{lm:t-wise-conn} that sampling edges of a $k$-connected graph at rate $\gamma \log n/k$  yields a connected graph with high probability, {\em even with limited independence}.  In particular, it is sufficient to ensure that random variables used for sampling edges incident to any given vertex are only $\tilde O(1/\e^2)$-wise independent.  This lemma is used in Section~\ref{sec:unlim-indep} to show that our estimation of sampling rates is accurate. We then show that our algorithm for constructing a sparsifier by sampling at rates proportional to edge connectivities yields a sparsifier with high probability even when the sampling is done using limited independence. We note that the second claim does not subsume the first due an extra $\log n$ factor that is needed for sampling with edge connectivities to go through. 

We will use tail bounds for $t$-wise independent random variables proved in \cite{SSS95}, Theorem 5:
\begin{theorem}\label{thm:ch}
Let $X_1,\ldots, X_n$ be   random variables each of which is confined to $[0, 1]$. Let $X=\sum_{i=1}^n X_i, \mu=\expect[X]$. Let $p=\mu/n$, and suppose that $p\leq 1/2$. Then if $X_i$ are $\lceil \e \mu \rceil$-wise independent, then 
$$
\prob[|X-\mu|\geq \e \mu]<e^{-\lfloor \e^2 \mu/3\rfloor}, 
$$
if $\e<1$, and 
$$
\prob[|X-\mu|\geq \e \mu]<e^{-\e \ln(1+\e) \mu/2}<e^{-\e \mu/3} 
$$
otherwise.
\end{theorem}

We now prove
\begin{lemma}\label{lm:t-wise-conn}
Let $G=(V, E)$ be a $k$-connected graph on $n$ nodes. For edges $e=(u, v)\in E$ let random numbers $h_{u, v}\in [0, 1]$ be such that 
\begin{enumerate}
\item $h_{u, v}$ is independent of $h_{u', v'}$ for all $u'\neq u$;
\item $h_{u, v}$ are $\lceil 4\gamma \log n\rceil$-wise independent for fixed $u$, where $\gamma\geq 20$.
\end{enumerate}
Also, let $X_{u, v}$ be $0/1$ random variables such that $X_{u, v}=1$ if $h_{u, v}\leq (\gamma \log n)/k$ and $0$ otherwise.
If $G'$ is obtained by including each edge $(u, v)\in E$ such that $X_{u, v}=1$ or $X_{v, u}=1$, then $G'$ is connected whp.
\end{lemma}
\begin{proof}
First, for each $e=(u, v)\in E$ let $\hat X_{u, v}$ denote $0/1$ random variables such that $X_{u, v}=1$ if $h_{u, v}<(\gamma  \log n)/s_e$, where $s_e$ is the strong connectivity of $e$. Define $\hat G'$ as the graph obtained by including each edge $(u, v)\in E$ such that $\hat X_{u, v}=1$ or $\hat X_{v, u}=1$. Note that $\hat G'$ is a subgraph of $G'$, so it is sufficient to show that $\hat G'$ is connected whp.

%We first define a useful decomposition, which we used in the recovery procedure in Algorithm~\ref{alg:partition}.  
Suppose that $F=(V_F, E_F)$ is a $k$-connected graph without $2k$-strongly connected components for some $k$. 
Recall from Lemma~\ref{lm:w-partition} that $V_F$ can be partitioned into $\log |V_F|$ sets $W_1,\ldots, W_{\log |V_F|}$ such that the degree of any $u \in W_r$ in $W_r \cup \cdots \cup W_{\log |V_F|}$ is at most $4k$.
%Let $H_1=V_F$. Repeat the following for $r=1,\ldots, O(\log n)$. Since $H_r$ does not have any $2k$-strongly connected components, then $$
%\sum_{e\in E(H_r)} w_e\leq 2k(|V(H_r)|-1)
%$$
%by Lemma~\ref{lm:k-weak}. Thus, at least half the nodes in $H_r$ have degree at most $4k$ (here by degree we mean sum of weights of incident edges). Denote the set of these nodes by $W_r$.
For each node $u\in W_r$ let $E_u$ denote the edges incident on $u$ that go to nodes in $W_r\cup\cdots\cup W_{\log |V_F|}$ (if an edge $e=(u, v)$ goes between two nodes in $W_r$, include it either in $E_u$ or $E_v$ arbitrarily). We will say vertex $u\in W_r$ {\em controls} edges $e\in E_u$. Note that $|E_u|\leq 4k$ for all $u\in W_r$, and hence the expected number of edges sampled in $E_u$ is at most $4\gamma \log n$. 
%Set $H_{r+1}=H_r\setminus W_r$. 
Note that this definition of control is slightly different from the one given in Definition~\ref{def:control}. In particular, this is because Definition~\ref{def:control} pertains to the actual sampling procedure that our algorithm uses, while here we are concerned with the estimation step.

Let $j_{max}=\lfloor \log_2 n\rfloor$. We will show by induction on $j=j_{max},\ldots, 0$ that all $2^j$-connected components are connected with probability at least $1-(j_{max}-j+1)n^{-3}$.

\begin{description}
\item[Base:$j=j_{max}$]  We have $\kappa=2^{j_{max}}$.  Apply the decomposition above to the $\kappa$-strongly connected components of $G$, which  does not have any $2 \kappa$-connected components since $2\kappa>n$.
Let $\hat G''$ denote the subgraph of $G$ obtained by including for each $u\in V$ edges $e=(u, v)\in E_u$ when $X_{u, v}=1$.
Denote the set of sampled edges by $E'$.
Recall that for all $u\in V(G)$ one has $\expect[|E'\cap E_u|]\leq 4\gamma \log n$.

Fix a cut $(C, V\setminus C)$. For each vertex $u\in C$  let $X_u=\sum_{(u, v)\in E_u, v \not \in C} X_{u, v}$.  

By setting $\e=1$ in Theorem~\ref{thm:ch} we get
$$
\prob[X_u=0]<e^{-\expect[X_u]/3}.
$$

Since $X_u, X_{u'}$ are independent for $u\neq u'$, the probability that the cut is empty is at most 
\[
\prod_{u\in C}e^{-\expect[X_u]/3}=e^{-\gamma |C|\log n/(3k)}.
\]
By Karger's cut counting lemma, the number of cuts of value at most $\alpha k$ is at most $n^{2\alpha}$.
 A union bound over all cuts, we get failure probability at most 
$$
\sum_{\alpha\geq 1} n^{2\alpha} e^{-\gamma \alpha \log n/3}\leq n^{-4}
$$
since $\gamma\geq 20$. Taking a union bound over all $\kappa$-connected components yields failure probability at most $n^{-3}$.

\item[Inductive step: $j+1\to j$] We have $\kappa=2^{j}$. By the inductive hypothesis, all $2^{j+1}$-connected components will be connected with probability at least $1-(j_{max}-(j+1)+1) n^{-3}$. We condition on this event and contract the connected components into supernodes. 

We now have a union of vertex-disjoint $\kappa$-strongly connected components that do not contain any $2\kappa$-connected components. The same argument as in the base case shows that each such component will be connected with probability at least $1-n^{-4}$. A union bound over at most $n$ such components completes the inductive step.
\end{description}

\end{proof}

In order to show that the results of \cite{FHHP11} carry over to our setting, it is sufficient to show that the following version of Chernoff bounds holds under our limited independence assumptions (Theorem 2.2 in \cite{FHHP11}):
\begin{theorem}\label{thm:ch-fhhp11}
Let $X_1,\ldots, X_n$ be $n$ random variables such that $X_i$ takes value $1/p_i$ with probability $p_i$ and $0$ otherwise.Then, for any $p$ such that $p\leq p_i$, for each $i$, any $\e\in (0, 1)$ and any $N\geq n$ the following holds:
$$
\prob\left[\left|\sum_{i=1}^n X_i-n\right|>\e N\right]<2e^{-0.38\e^2 p N}.
$$
\end{theorem}

Indeed, an inspection of the proofs of  Lemma~4.1 and Lemma~5.5 in \cite{FHHP11} shows that the authors (a) only rely on independence of their sampling process to obtain Theorem~\ref{thm:ch-fhhp11} and (b) only apply Theorem~\ref{thm:ch-fhhp11} to subsets of edges of $G$, where $p_i$ are sampling probabilities.Thus, proving an equivalent of Theorem~\ref{thm:ch-fhhp11} allows us to extend their results to our limited independence sampling approach. 

\newcommand{\X}{{\cal X}}
\newcommand{\E}{{\cal E}}

We now prove
\begin{lemma}\label{lm:ch-1}
Let $G=(V, E)$ denote an unweighted undirected graph. Let $\gamma>0$ be a sufficiently large constant such that sampling at rate $\gamma \log^2 n/c_e$ independently produces a sparsifier with probability at least $1-n^{-2}$, where $c_e$ is the edge connectivity of $e$. Let $X_e, e\in E$ be random variables corresponding to including edges from a set $E^*$ into the sample such that $X_e$ takes value $1/p_e$ with probability $p_e$ and $0$ otherwise, where $p_e$ is the sampling probability used by Algorithm~\ref{alg:rec-edges}. Assume that sampling is $c\log^4 n/\e^2$-wise independent for a sufficiently large constant $c>0$ that may depend on $\gamma$.

There exists an event $\E$ with $\prob[\E]>1-n^{-2}$  such that for any $E^*\subseteq E$, any  $p\leq p_e, e\in E^*$,  any $\e\in (0, 1)$ and any $N\geq |E^*|$ 
$$
\prob\left[\left|\sum_{e\in E^*} X_e-|E^*|\right|>\e N|\E\right]<e^{-\e^2 p N/6}.
$$

\end{lemma}
\begin{proof}
For simplicity of exposition, we now assume that $G$ is unweighted.
Let $X_1,\ldots, X_n$ be random variables corresponding to picking edges of the graph. Recall that our sampling algorithm samples an edge $(u, v)$ either depending on the value of $g_{(u, v)}$ or the value of $g_{(v,u)}$ (the choice depends on the partition of the node set $U_1\cup\ldots\cup U_r$ constructed in Algorithm~\ref{alg:rec-edges}). Recall that by Definition~\ref{def:control} a node $u$ {\em controls} edge $(u, v)$ if Algorithm~\ref{alg:rec-edges} samples $(u, v)$ using the value of $g^*_{(u, v)}$. Note that each edge is controlled by exactly one node. For a node $u$, as before, let $E_u$ denote the set of edges controlled by $u$. By Lemma~\ref{lm:control-small}, one has $\expect[|E^*\cap E_u|]=O(\log^4 n/\e^2)$. Let $\E$ denote the event that at most $2\gamma \log^4 n/\e^2$ edges controlled by $u$ are sampled, for all $u\in V$, where we are assuming that $\gamma$ is sufficiently large. Since our random variables are $c\log^4 n/\e^2$-wise independent for sufficiently large $c$, by Theorem~\ref{thm:ch} and a union bound over all $u$ one has $\prob[\E]\geq 1-n^{-2}$.

For each $e\in E$ let $X_{e}$ be a Bernoulli random variable that takes value $p/p_e$ if edge $e$ is sampled, and $0$ otherwise, so that $X_e\in [0, 1]$.

Consider a set of edges $E^*$.  Partition $E^*$ as $E^*=\bigcup_{u\in V} E^*_u$, where $E^*_u=E^*\cap E_u$. Thus,
 random variables $\X_u:=\sum_{e\in E^*_u} X_e$  are independent for different $i$. Let $\X=\sum_{u\in V} \X_u$, $\mu=\expect[\X]$.

Then by Markov's inequality
\begin{equation}
\prob[\X\geq (1+\delta) \mu|\E]\leq \frac{\expect[e^{t\X}|\E]}{e^{t(1+\delta) \mu}}.
\end{equation}

Recall that 
\begin{equation}\label{eq:series}
\expect[e^{t\X}|\E]=\sum_{j=0}^{\infty}\expect[(t\X)^j|\E]/j!=\sum_{j=0}^{\infty} t^j/j!\sum_{S\subseteq E^*, |S|\leq j}\sum_{\alpha_e\geq 0, \sum_{e\in S}\alpha_e=j} \expect\left[\prod_{e\in S} X_e^{\alpha_e}|\E\right]
\end{equation}
\newcommand{\Y}{{\cal Y}}
For any non-negative random variably $\Y$ one has 
$$
\expect[\Y|\E]\leq \expect[\Y]/\prob[\E].
$$

\if 0
$$
\expect[\Y|\E]\prob[\E]+\expect[\Y|\bar \E]\prob[\bar \E]=\expect[\Y],
$$
so 
$$
\expect[\Y|\E]=(\expect[\Y]-\expect[\Y|\bar \E]\prob[\bar \E])/\prob[\E]\leq \expect[\Y]/\prob[\E].
$$
\fi

Conditional on $\E$, one has $\prod_{e\in S} X_e^{\alpha_e}=0$ for all $S\subseteq E^*$ such that $|S\cap E_u|>2\gamma \log^4 n/\e^2$ for at least one $u\in V$. For other $S$, setting  $\Y=\prod_{e\in S} X_e^{\alpha_e}$,  one gets
\begin{equation}\label{eq:uncond}
\expect\left[\prod_{e\in S} X_e^{\alpha_e}|\E\right]\leq \expect\left[\prod_{e\in S} X_e^{\alpha_e}\right]/\prob[\E].
\end{equation}

Combining \eqref{eq:uncond} and \eqref{eq:series} one gets
\begin{equation}
\expect[e^{t\X}|\E]\leq \frac1{\prob[\E]}\sum_{j=0}^{\infty} t^j/j!\sum_{S\subseteq E^*, |S|\leq j, |S\cap E_u|\leq 2\gamma \log^4 n/\e^2}\sum_{\alpha_e\geq 0, \sum_{e\in S}\alpha_e=j} \expect\left[\prod_{e\in S} X_e^{\alpha_e}\right]
\end{equation}

On the other hand, for all $S\subseteq E^*$ such that $|S\cap E_u|\leq 2\gamma \log^4 n/\e^2$ one has 
$$
\expect\left[\prod_{e\in S} X_e^{\alpha_e}\right]=\prod_{e\in S} \expect[X_e^{\alpha_e}]
$$
by $\gamma \log^4 n/\e^2$-wise independence. Thus, we get 
\begin{equation}
\prob[\X\geq (1+\delta) \mu|\E]\leq \frac1{\prob[\E]}\frac{\prod_{e\in E^*}\expect[e^{t X_e}]}{e^{t(1+\delta) \mu}},
\end{equation}
which is the same bound as in the full independence case, except for a factor of $1/\prob[\E]=1+O(1/n)$ in front. Now the same derivation as in the full independence case  shows that the probability of overestimating is appropriately small.

We now bound the probability of underestimating.
Consider a set of edges $E^*$.  Partition $E^*$ as $E^*=\bigcup_{i=1}^s E_i$, where $E_i\cap E_j=\emptyset, i\neq j$, so that 
\begin{enumerate}
\item $\expect[\sum_{e\in E_i} X_e] \leq c \log^4 n/\e^2$ for a sufficiently large constant $c>0$;
\item random variables $\sum_{e\in E_i} X_e$  are independent for different $i$;
\item $s\leq \frac{\e^2}{6\log (4/\e^2)}\expect[\sum_{e\in E^*}X_e]$. 
\end{enumerate}
Note that this is feasible since our graphs are unweighted, so $\e$ can be assumed to be larger than $1/n^2$.
 For each $i=1,\ldots, s$ let $\X_i:=\sum_{e\in E_i} X_{e}$. Note that $\X_i$ are independent, and $X_e$ are $c\log^4 n/\e^2$-wise independent. 
Now by Theorem~\ref{thm:ch} for all $i=1,\ldots, s$ one has for all $\e\in (0, 1)$ 
\begin{equation}\label{eq:2}
\prob[\X_i<\expect[\X_i]-\e \expect[\X_i]]<e^{-\e^2 \expect[\X_i]/3}
\end{equation}

\newcommand{\K}{{\cal K}}
\newcommand{\z}{{\mathbf{z}}}

Let $\X=\sum_{i=1}^s \X_i$. 
For constant $\e>0$ let 
$$
\K(\e)=\left\{\z=(z_1,z_2,\ldots,z_s)\in \left\{0, \frac1{4}\e^2, \frac1{2}\e^2, \frac{3}{4} \e^2,\ldots, 1-\e^2/4, 1\right\}^s: \sum_{i=1}^s z_i \expect[\X_i] \geq \e \expect[\X]\right\}.
$$

 We now have 
\begin{equation}
\begin{split}
\prob\left[\X<\expect\left[\X\right]-\e \expect[\E]\right]&\leq \sum_{\z\in \K(\e)} \prod_{i=1}^s\prob[\X_i<\expect[\X_i]-(z_i-\e^2/4)^2\expect[\X_i]]\\
&\leq \sum_{\z\in \K(\e)} \prod_{i=1}^s\prob[\X_i<\expect[\X_i]-z_i^2\expect[\X_i]+(\e^2/2)\expect[\X_i]]
\end{split}
\end{equation}
since every set of values for $\X_i-\expect[\X_i]$ such that $\sum_{i} (\X_i-\expect[\X_i])<-\e \expect[\X]$ can be rounded to a point in $\K(\e)$ with a loss of at most $(\e^2/2) \expect[\X_i]$ in each term.
We now note that  for any $\z\in [0, 1]^s$ such that $\sum_{i=1}^s z_i \expect[\X_i]=\e' \expect[\X]\geq \e \expect[\X]$ one has $\sum_{i=1}^s z_i^2 \expect[\X_i]\geq (\e')^2 \expect[\X]\geq \e^2 \expect[\X]$.

Next, since $s\leq \frac{\e^2}{6\log(4/\e^2)} (\sum_{i=1}^s \expect[\X_i])$, we have that
\begin{equation}
\prob\left[\X<\expect\left[\X\right]-\e \expect[\X]\right]\leq (4/\e^2)^{s} e^{-\e^2 \expect[\X]/3}\leq e^{-\e^2 \expect[\X]/6}.
\end{equation}
\end{proof}

\begin{theorem}
The set of edges returned by Algorithm~\ref{alg:rec-edges} is a sparsifier whp.
\end{theorem}
\begin{proof}
Lemma~\ref{lm:ch-1} can be used instead of Theorem~\ref{thm:ch-fhhp11} in \cite{FHHP11}.
\end{proof}

\pdfbookmark[1]{\refname}{My\refname} \bibliographystyle{alphaurl}
%\bibliography{sparsification}
\newcommand{\etalchar}[1]{$^{#1}$}

\end{document}